\documentclass{amsart}

\usepackage{amssymb}
\usepackage{enumerate}
\usepackage{amsmath}
\usepackage{amsthm}
\usepackage{textcomp}

\usepackage{tikz}
\usetikzlibrary{plotmarks, shapes.geometric}

\usepackage{graphicx,epsfig,tikz}
\definecolor{wrwrwr}{rgb}{0.3803921568627451,0.3803921568627451,0.3803921568627451}
\definecolor{dtsfsf}{rgb}{0.8274509803921568,0.1843137254901961,0.1843137254901961}
\definecolor{sexdts}{rgb}{0.1803921568627451,0.49019607843137253,0.19607843137254902}
\definecolor{rvwvcq}{rgb}{0.08235294117647059,0.396078431372549,0.7529411764705882}
\definecolor{cqcqcq}{rgb}{0.7529411764705882,0.7529411764705882,0.7529411764705882}

\usepackage{amsmath,amsfonts,amssymb,enumerate,mdwlist}
\usepackage{latexsym,amsfonts,amssymb,amscd}
\usepackage{fancybox,shadow,color,graphicx,psfrag,float}

\usepackage{geometry}
\usepackage{stackengine}

\newcommand{\F}{\mathbb{F}}
\newcommand{\C}{{C}}
\renewcommand{\P}{\mathbb{P}}

\renewcommand{\textcolor}[2]{#2} 

\newtheorem{thm}{Theorem}[section]
\newtheorem{pro}[thm]{Proposition}
\newtheorem{lem}[thm]{Lemma}
\newtheorem{cor}[thm]{Corollary}

\newtheorem{rem}[thm]{Remark}

\theoremstyle{definition}
\newtheorem{defin}{Definition}

\begin{document}

\title{ LRC codes over characteristic $2$} %

\author{ F. Galluccio } 

     \email{F. Galluccio: fgalluccio @ fiq.unl.edu.ar}

\thanks{\\This work was partially supported by UNL CAI+D 2024 and a doctoral grant by CONICET. The author wishes to express his gratitude to his supervisors: María Chara, for her essential contributions to the drafting and writing of this article, and Edgar Martínez-Moro for his review and comments. \\ 
F. Galluccio: Universidad Nacional del Litoral and CONICET, Santa Fe, Argentina }

\begin{abstract}

In this work the construction of LRC codes given in \cite{lrcontowers} is completed, in the case of even characteristic.
A general construction is presented, that enables us to obtain linear LRC codes of large length $n \approx q^4$, dimension and distance of order $q^4$, and locality $r =q-1$. 
In addition, the cases $q = 4$ and $q=8$ are studied.


\end{abstract}
\maketitle

\section{Introduction}

 Let $q = p^l$ be a power of a prime.
An $\textcolor{blue}{[n,k,d]}$ linear code $C$ is a $k$-dimensional subspace of $\F_q^n$ with minimum distance $d$.
In this context, for linear codes, the minimum distance is defined as
$$d = \min \left\{ d(x,0): \quad 0 \neq x \in C \right\},$$
where $0 \leq d(x,y) \leq n$ denotes the Hamming distance, i.e.\ the number of components in which the two vectors $x, y \in \F_q^n$ differ.
Consequently, $d(x,0)$ represents the number of nonzero coordinates in the vector $x$.
 
The importance of error-correcting codes lies in their ability to \textit{detect} and \textit{correct} errors that may arise due to transmission issues or loss of access to certain information.
That is, if a vector $x \in C \subset \F_q^n$ is sent through a channel and a vector $y \in \F_q^n$ is received, the decoder aims to reliably recover the original message vector $x$ using the structure of the code.

Formally, let $C$ be an $\textcolor{blue}{[n,k,d]}$ linear code, and let $t$ be such that $d \geq 2t + 1$. Then we say that $C$ can detect and correct up to $t$ errors using nearest-neighbor decoding:
For vectors $x \in C$ and $y \in \F_q^n$ with $d(x,y) \leq t$, then $x$ is the unique codeword in $C$ within distance at most $t$ from $y$.
Using this decoding, we obtain that for any codeword $x \in C$, and a set of indexes $i_1 < i_2 < \ldots < i_t$, with $2t + 1 \leq d$, it is possible to recover the coordinates $\{x_{i_1}, \ldots, x_{i_t}\}$ from the remaining $n - t$ coordinates.

The central question that motivates the study of locally recoverable codes is whether it is possible to recover any specific coordinate $x_i$ of the original message from a subset of $r$ coordinates, where $r$ is significantly smaller than $n - t$.

Following the definitions given in \cite{bargtamo}, we have that

\begin{defin}

An $\textcolor{blue}{[n,k,d]}$ code $C \subset \F_q^n$ is \emph{locally recoverable with locality $r$} if for every $1 \leq i \leq n$ there exists a subset $A_i \subset [1,n] \setminus \{i\}$, $|A_i| = r$, and a function $\phi_i : \F_q^{|A_i|} \to \F_q$ such that for every codeword $x = (x_1, x_2, \ldots, x_n) \in C$ we have
$$x_i = \phi_i(x_{j_1}, \dots, x_{j_{|A_i|}}),$$
\end{defin}

Note that the function $\phi_i$ in the previous definition depends only on the code $C$ and the position $i$, but not on the codeword $x$.
Moreover, there is no requirement for $\phi_i$ to be an algebraic function.
 
The aim of this work is to complete the constructions given in \cite{lrcontowers} by constructing examples of Locally Recoverable Codes (LRC Codes) defined over finite fields of even characteristic \textcolor{blue}{providing explicit examples and explicitly computing their dimension and minimum distance.}
In that article, there were obtained collections of LRC Codes with locality $r$ over $\F_{q^2}$\textcolor{blue}{, for odd $q$,} through the use of towers of function fields, in a manner that $r$ only depended on the degrees of each extension within the tower, whereas that their length $n_i$ is arbitrarily large, and $\frac{d_i}{n_i} + \frac{k_i}{n_i} \rightarrow \delta + R \approx 1$.

\textcolor{blue}{In the previous mentioned work, and in some references therein, some examples of} LRC codes with locality $r$ \textcolor{blue}{were exhibited} where the dimension $k$ is bounded and the minimum distance $d$ is known, or, viceversa, where the dimension $k$ is known and the minimum distance is bounded.
\textcolor{blue}{Other interesting examples} were obtained via Automorphisms in \cite{lrcautomorphisms} or via cartesian product of curves in \cite{hermitianlifted}. 
\textcolor{blue}{In particular, examples constructed over the rational function field were described in \cite{automorphismsrational} and \cite{automorphismsmultiple}, and over the Garcia Stichtenoth tower (see \cite{garciaetal}) were described in \cite{goodautomorphisms}.}

\textcolor{blue}{In the majority of the examples of this present} work, the dimension of the code will be known and can be computed, and the distance $d$ will be explicitly determined, while providing more insight of the vector space considered for the construction. It will be considered $\mathcal{F} = \left( F_0, F_1, F_2, \ldots \right)$ a tower of function fields (c.f.
section \ref{prelim}) such that $F_0  = \F_q(x_0)$ and $F_i = F_{i-1}(x_i)$ with $f(x_{i-1}, x_{i}) = 0$ for some suitable polynomial $f(x,y)$ and the following version of the construction from   \cite[Theorem 3.4]{lrcontowers}, \textcolor{blue}{whose proof can be obtained by adapting the arguments in the cited article}.
 
\begin{thm}\label{teoremareversionado}
Suppose $S \subset F_0$ be a set of $s$ rational places that split completely in rational places of $F_j$ that form a set $\mathcal{B}$ of $m \cdot s$ elements.
If 
$$V = 
\left \langle x_0^{e_0} \cdots x_j^{e_j}: \: x_i \in F_i, \: \textcolor{blue}{0 \leq e_0 \leq l, \:} 0 \leq e_i \leq \textcolor{blue}{m_i-1}, \: 0 \leq e_j \leq m_j-2 \right\rangle \subset F_j,$$ \textcolor{blue}{for $1 \leq i \leq j-1$,}
is an adequate vector space (in a manner that no $P \in \mathcal{B}$ is a pole of any $f \in V$), then the linear code \textcolor{blue}{$C = C_i(\mathcal{B}, V)$ is an LRC} Code of length $n = ms$, dimension $k = \textcolor{blue}{(l+1) m_1 \cdots m_{j-1}} (m_j-1) $, locality $r = m_j-1$ and minimum distance 
$$d \geq n - \textcolor{blue}{l \deg(x_0)}- \sum_{i= \textcolor{blue}{1}}^{j-1} \textcolor{blue}{(m_i-1)}\deg(x_i) - (m_j - 2) \deg(x_j).$$

\textcolor{blue}{After establishing any ordering of the elements of $\mathcal{B}$, we can recover any missing coordinate $f(P_i)$, for $P_i \in \mathcal{B}$ with $1 \leq i \leq n = q^i(q^2-q)$, by defining 
$$A_i = \left\{ P \in \mathcal{B}: \: P \cap F_{j-1} = P_i \cap F_{j-1}\right\} \setminus\{P_i\}$$ 
of cardinality $[F_j:F_{j-1}]-1 = m_j-1$.
 In this case, any element $f \in V$ restricted to $A_i$ behaves as a polynomial in the function $x_j$ of degree at most $m_j-2$ since the elements $x_0, x_1, \ldots, x_{j-1}$ are constant functions in these sets $A_i$.
 So it is possible to recover $f(P_i)$ as the evaluation $\widetilde{f}(x_j(P_i))$ where $\widetilde{f}(T) = \sum_{i=0}^{m_j-1} b_i T^i$ is obtained by polynomial interpolation on the points $\{(x_j(P), f(P)), \: P \in A_i\}$. }
\end{thm}

It is worth to note here that $d \geq 0$ is a necessary condition for this code to exist, while $n - \textcolor{blue}{l \deg(x_0)}- \sum_{i= \textcolor{blue}{1}}^{j-1} \textcolor{blue}{(m_i-1)}\deg(x_i) - (m_j - 2) \deg(x_j) \geq 0$ is a sufficient condition.
We still do not know any necessary and sufficient conditions for $d \geq 0$ independently of the choice of the tower $\mathcal{F}$ nor the base field $F_0$.

Nevertheless, for any suitable tower $\mathcal{F}$, set $S\subset F_0$ of rational places that split completely on $\mathcal{F}$ and vector space $V = \left \langle x_0^{e_0}\cdots x_j^{e_j} , \: 0 \leq e_i \leq k_i \right \rangle,$ we can construct LRC Codes of length $n = m \cdot s$, dimension $k = \dim V= k_0 \cdots k_j $ and \textcolor{blue}{known minimum distance $d = n - \sum_{i=0}^j k_i \deg(x_i)$.} 

These constructions will be exhibited in the sections \ref{primera}, specifically \ref{ej313}, and \ref{segunda}.
Later in \ref{conclusiones} these constructions will be compared.

\section{Preliminaries}\label{prelim}

In this section we will proceed to briefly define terms and objects which will be used among this work.
By nothing that any linear code $C$ is a subspace of $\F_q^n$, we can obtain examples of linear codes as images of injective morphisms $f: \F_q^k \to \F_q^n$.
These definitions may be revisited with more detail in \cite{stichtenoth}.

\begin{defin} For a function field $F / \F$, let $\P_F$ the set of all its places.
Let $P_1, \ldots, P_k$ be places of a function field $F/ \F$, and $n_1, \ldots, n_k \in \mathbb{Z}$.
A divisor of $F/\F$ is a (finite) formal sum $D = n_1P_1 + \ldots + n_kP_k = \sum_{\P_F} n_P P$.
The Riemann-Roch space $\mathcal{L}(D)$ is the $\F-$vector space of all functions $f \in F$ such that $v_P(f) \geq -n_P$ for all $P \in \P_F$.
The support of $D$ is the set $\{P: \: n_P \neq 0 \}$, its degree is $\deg(D) = \sum n_P$, and its dimension is $\ell(D) = \dim\left( \mathcal{L} (D) \right).$
\end{defin}

\begin{defin}[AG Code]
Let $D = P_1+ P_2 + \ldots + P_n$ be a divisor whose support consists of $n$ distinct rational places of $F$, and let $G$ be any divisor of $F / \F_q$ whose support does not contain any $P_i$.
Suppose $k = \ell(G) - \ell(G-D) > 0$ and $d = n - \deg(G) > 0$.
The Algebraic-Geometric Code (or simply AG Code) is defined as
$$C = C_{\mathcal{L}}(D,G) = \left\{(x(P_1), x(P_2), \ldots, x(P_n)) : \: x \in \mathcal{L}(G) \right\} \subset \F_q^n.$$
\end{defin}

We may define similarly an evaluation code as in the definition of AG Code.
It is important to note that not all Evaluation Codes will be AG Codes, neither all AG Codes will be Evaluation Codes.

\begin{defin}[Evaluation Code]
Let $B = \{P_1, P_2, \ldots, P_n \}$ be a set of $n$ distinct rational places of $F$, and let $G$ be a divisor of $F / \F$ whose support does not contain any $P_i$.
Suppose $V \subset \mathcal{L}(G)$ verifies $k = \dim V > 0$ and  $d = n - \deg(G) > 0$.

The evaluation code is defined as 
$$C = C(B,V) = \left\{(x(P_1), x(P_2), \ldots, x(P_n)) : \: x \in V \subset \mathcal{L}(G) \right\}.$$
\end{defin}

Finally, we present the definition of towers of function fields.
In the case of good asymptotic behavior of these towers we will have several examples of function fields with many rational places and vector spaces with known dimension.

\begin{defin}[Tower of function fields]
Let $F_0 \subset F_1 \subset F_2 \subset \cdots$ be a sequence of function fields over $\F_q$ such that every extension $F_{i+1}/F_i$ is finite and separable, the sequence $\{g(F_i)\}$ of their genus goes to $\infty$, and $F_i \neq F_{i+1}$.
Suppose also that every element in $F_i$ algebraic over $\F_q$ is also on $\F_q$, for each $i$.
Then we say that $\mathcal{F} = (F_0, F_1, F_2, \ldots)$ is a tower of function fields over $\F_q$.
\end{defin}

A tower $\mathcal{E} = (E_0, E_1, E_2, \ldots)$ is a subtower of another tower $\mathcal{F}$ if for each $i$ there is an embedding $E_i \subseteq F_j$ for some $j$.

In particular, a useful way to obtain towers of function fields is to consider a polynomial $f(x,y)$ of degree $d \geq 2$ in each variable, and define recursively $F_i = F_{i-1}(x_i)$ for $i\geq 1$ where $F_0 = \F_q(x_0)$ for $x_0$ trascendental over $\F_q$, and $x_{i}$ and $x_{i-1}$ are algebraically dependent over $\F_q$ via the relation $f(x_{i-1}, x_i) = 0$.
In this particular case we can understand explicitly the extension $F_i/F_{i-1}$.
Moreover,  a necessary condition for a good asymptotic behavior is equal degree in each variable in the polynomial $f(x,y)$ (see \cite{skewpyramids}).

\section{Construction 1} \label{primera}

Consider the extensions $F_0 \subset F_1 \subset F_2$ belonging to the tower of function fields described in \cite{garciaetal},
$$\mathcal{W} = (F_0, F_1, F_2, \ldots),$$
defined recursively over $\F_{q^2}$, with $q=2^{\ell} > 4$, by $F_0 = \F_{q^2}(x_0)$ and $F_i = F_{i-1}(x_i)$ where
$$x_{i+1}^q + x_{i+1} = \dfrac{ x_{i}^q}{x_i^{q-1}+1}.$$

We can characterize the rational places that split completely in $\mathcal{W}$ in the following sense: If 
$$S_0 = \F_{q^2} \setminus \F_q = \{ \beta \in \F_{q^2}: \beta^q + \beta \neq 0 \},$$
we have that the set of rational places in $R_0$ that splits completely in $\mathcal{W}$ is
$$S = \{ P_{\beta} \in \P(F_0): \: \beta \in S_0 \}.$$
In fact, it is easy to see that $P_1$ ramifies in the extension $F_1/F_0$, while $P_{0}$ splits in the $q$ places $P_{0\beta}$ with $\beta \in \F_q$ and the place $P_{01}$ of $F_1$ ramifies in $F_2/F_1$.
Since the extensions are defined recursively, we have that for $\beta \in \F_{q^2} \setminus \{0,1\}$, the place $P_{\beta}$ splits totally in the extension $F_2/F_0$ if and only if  $\frac{x_0(P_{\beta})^q}{x_0(P_{\beta})^q+1} \neq 0,1$, and since $\beta = x_0(P_{\beta}) \neq 0$, it is equivalente  to $\beta^q+\beta = Tr(\beta) \neq 0$.

\textcolor{blue}{In this context, $Tr: \F_{q^2} \to \F_q$ is the trace function, with $Tr(\beta) = \beta^q + \beta$, and $N: \F_{q^2} \to \F_q$ is the norm function, with $N(\beta) = \beta \cdot \beta^q = \beta^{q+1}$.}

\begin{pro}\label{propdecompo}
The $q^2-q$ rational places of $S$ can be naturally partitioned in $(q-1)$ subsets $S_1, \ldots, S_{q-1}$, each one of $q$ elements, characterized by the value $$b_i = \frac{\beta^{q+1}}{\beta^q+\beta} = \frac{N(\beta)}{Tr(\beta)} \in \F_q^*,$$\textcolor{blue}{$\beta \in S_0$}.
Each place $P$, with $x_0(P) \in S_i$, splits completely in $q$ rational places $Q \mid P_{\beta}$ with $x_1(Q) = \gamma$ and $Tr(\gamma) = b_i.$
\end{pro}

In order to prove this proposition, it is convenient to recall the following result:

\begin{lem} \label{lema3} Let $\F_q^* = \F_q \setminus\{0 \}$:
\begin{enumerate} \item given $b \in \F_q^*$, there exist exactly $q$ elements $\gamma \in \F_{q^2}$ with $Tr(\gamma) = b$,
\item given $c \in \F_q^*$, there exist exactly $q+1$ elements $\gamma \in \F_{q^2}$ with $N(\gamma) = c$,
\item given $b,c \in \F_q^*$ there exist either $0$ or $2$ elements $\gamma \in \F_{q^2}$ with $Tr(\gamma)=b$ and $N(\gamma)=c$.
\end{enumerate}
\end{lem}

\begin{proof}
 \phantom{a}

\begin{enumerate} 
\item It is clear that the equation $\gamma^q+\gamma=b$, \textcolor{blue}{where $\gamma$ is a variable}, has at most $q$ solutions in $\F_{q^2}$.
By the pigeonhole principle, since there are $q^2-q$ elements in $\F_q^2$ with nonzero trace, and exactly $q-1$ possible values for its trace, necessarily each equation has exactly $q$ solutions.

 \item It is clear that the equation $\gamma^{q+1}=c$, \textcolor{blue}{where $\gamma$ is a variable}, has at most $q+1$ solutions in $\F_{q^2}$.
However,  any $a \in \mathbb{F}_q$ verifies $a^{q+1} = a^qa = a^2$.
Since every element in a finite field of even characteristic is a square, there is exactly one solution in $\F_q$, and thus any other solution must be on $\F_{q^2}$, hence in $S_0$.
By the pigeonhole principle, since there are $q^2-q$ elements in $S_0$, and exactly $q-1$ possible values for its norm, necessarily each equation has exactly $q$ additional solutions in $S_0$.
 
 \item Given any solution $\gamma \in \F_{q^2}$ \textcolor{blue}{to both equations in the previous items}, there exists $\gamma' = \gamma^q$ with $Tr(\gamma) = \gamma + \gamma'$ and $N(\gamma) = \gamma \gamma'$.
Since the characteristic polynomial of $\gamma$ is precisely
$$p_{\gamma}(T) = T^2 - bT +c,$$
and since $b \neq 0$ we do have $\gamma \neq \gamma'$.
The existence of solutions is trivial, since $\F_{q^2} / \F_q $ is the only extension of degree $2$ of $\F_q$, either the polynomial is irreducible and has two roots in $\F_{q^2}$ or it factors and thus has two roots in $\F_q$.
   
 \end{enumerate}
\end{proof}

Now we can prove the Proposition \ref{propdecompo}.

\begin{proof}
Let $P_{\beta} \in S \subset \P(F_0)$, and let $\beta = x_0(P)$ which is non-zero.
For $i=1, \ldots, q-1$, consider $b_i =\frac{N(\beta)}{Tr(\beta)}$, and let $S_i$ be the set of all elements $\alpha \in S$ such that $\frac{N(\alpha)}{Tr(\alpha)} = b_i$.
Now for any $\alpha \in S_i$, and any $Q \mid P_{\alpha}$, it holds that $x_1(Q) = \gamma$ with $Tr(\gamma) = \frac{N(\alpha)}{Tr(\alpha)} = b_i$.
Since there are exactly $q$ solutions $\gamma$, and $q$ distinct places $Q \mid P$, the set $\{x_1(Q): Q \mid P_{\alpha} \}$ is in fact the same for all such $P$.
The partition is natural in the sense it does not depend on the choice of any generator for $\F_q$ nor $\F_{q^2}$.
\end{proof}

We can obtain the following:

\begin{pro}
Let $B_i= \{ \beta \in \F_{q^2}: \: Tr(\beta) = b_i \}$, for $b_i \in \F_q^*$.
Then for each pair of indexes $1 \leq i,k \leq q-1$, we have that
$$|S_k \cap B_i | = | \{ \alpha \in \F_{q^2} : \: Tr(\alpha) = b_i, \: N(\alpha) = b_ib_k \} | \leq 2,$$
so that, for each rational place $P \in \P(F_j)$ that splits completely, the set of places $Q \in \P(F_{j+1})$, $Q \mid P$ can be naturally partitioned into $\frac{q}{2}$ pairs of places such that each pair is characterized by a different value of $\dfrac{N( x_{j+1}(Q))}{Tr(x_{j+1}(Q))}$.
\end{pro}

\begin{proof}
For $\alpha \in S_k$, and setting $b_k = \frac{N(\alpha)}{Tr(\alpha)}$, if $\alpha \in B_i$ then $Tr(\alpha) = b_i \neq 0$ and then $N(\alpha)=b_ib_k \neq 0$, and by the previous lemma it is clear that there are either $0$ or $2$ such $\alpha$.

Now, for $P \in F_j$, we can show by induction that, if $x_j(P) \in S_k$ and $Q \mid P$, then $\gamma = x_{j+1}(Q)$ is a solution of the equation $\gamma^q+\gamma = b_k$ so $\gamma \in B_k$.
However, since $Tr(\gamma) \neq 0$ it must be that $\gamma \in S_{h}$ for some $1 \leq h \leq q-1$, and therefore $\gamma' = \gamma^q$ is also in $B_k \cap S_h$, which shows that there exists $Q' \mid P$ such that $x_{j+1}(Q') = \gamma' \in S_{h} \cap B_k$ .
\end{proof}

In the following theorem, we adapt the method presented in \cite{lrcontowers} to the case where q is even.

\begin{thm}  \label{ej312} Let $q > 2$ be an even prime power, and let $S = \{P_{\alpha} \in \P (F_0) : \alpha \in S_0 \}$ with $S_0 = \F_{q^2} \setminus \F_q$.
If
$$\mathcal{B} = \{Q \in \P(F_2): \: Q \mid P_{\alpha} \text{ for some }  P_{\alpha} \in S\}$$
 and $V$ is the vector space spanned by
$$\{x_0^i x_1^j x_2^k: 0 \leq i \leq \frac{q^2}{2}-q, \: 0 \leq j \leq q-1, \: 0 \leq k \leq q-2\},$$
then $V \subset \mathcal{L}(l P_{\infty})$ for some $l \in \mathbb{N}$ and the evaluation code \textcolor{blue}{$\C = \C_2(\mathcal{B}, V)$} is an $\textcolor{blue}{[n,k,d]}$ LRC code \textcolor{blue}{ of locality $r = q-1$ with parameters}: 

\begin{align*} n &=q^2(q^2-q),  \\
k & = \left( \frac{q^2}{2} - q +1 \right) q(q-1), \\
d & = n - \left(\frac{q^2}{2}-q \right) \cdot q^2 - (q-1)q^2 - (q-2)q^2,  \\ 
& = q^2\left(  \frac{q^2}{2} - 2q +3 \right).
\end{align*}
\end{thm}

\begin{rem}
Since the point at infinity in each $F_i$ is totally ramified in $F_{i+1}/F_i$ for each $i \geq 0$, we have that $P_{\infty} \in F_2$ is a pole of $x_0, x_1$ and $x_2$ with $\deg (x_0^i x_1^j x_2^k)_{\infty} \leq \left(\frac{q^2}{2}-q \right) \cdot q^2 + (q-1)q^2 + (q-2)q^2 $  so $V \subset \mathcal{L}(l P_{\infty})$ for $l \geq \left(\frac{q^2}{2}-q \right) \cdot q^2 + (q-1)q^2 + (q-2)q^2$.
\end{rem}

\begin{proof}

As in \cite{lrcontowers}, we know that $\deg(x_0) = \deg(x_1) = \deg(x_2) = q^2$, so that the expression for $d = n-l$ is effectively a lower bound.

We must show that there exist some function $f \in V$ with $n-d$ zeros in $\mathcal{B}$.

Consider $S_1 \subset S_0$ arbitrarily, which has $q$ elements.
It is easy to see $| \{ P \in \P(F_1) : \: x_1(P) \in S_1 \}| = q^2$ since, for each $\beta \in S_1$, there exist exactly $q$ places $P_{\alpha} \in \P(F_0)$ laying under some place $P$ with $x_1(P) = \beta$.
The last part of Lemma \ref{lema3} asserts that those $q^2$ places come \textit{in pairs}: by matching or identifying each set $S_i$ as one of $q-1$ different \textit{colors},  we see that for each $P_{\alpha} \in \P(F_0)$ there exist two places $P, \: P'$ of color $S_1$ according to the value $x_1(P)$.
In particular, $P' = \sigma(P)$ where $\sigma$ is the only non trivial automorphism of $\F_{q^2}/\F_q$, which is Galois since $[F_{q^2}:F_q]=2$ .

This pairing shows that, in consequence,  there exist exactly $\frac{q^2}{2}$ places $P_{\alpha} \in S$ with exactly $2$ places above it of the same color $S_1$.

Consider the remaining $q^2-q - \frac{q^2}{2}$ places, that is,
$$\mathbb{H}_0 = \{ P_{\alpha} \in \P(F_0) : \: x_1(P) \not\in S_1 \: \text{ for all } P|P_{\alpha}\}$$
of size $\frac{q^2}{2} - q$ and let
$$H_0  = \{x_0(P_{\alpha}): \: P_{\alpha} \in \mathbb{H}_0  \} = \{\alpha \in S_0 : \: P_{\alpha} \in \mathbb{H}_0 \}.$$

The function
$$h_0 =  \prod_{\alpha \in H_0} (x_0 - \alpha) \in V$$
has exactly $q^2 \cdot \left( \frac{q^2}{2} -q \right)$ zeros in $\mathcal{B}$, precisely those places $Q$ lying above some place in $\mathbb{H}_0$.

Among the remaining places $Q \in \mathcal{B}$ such that $Q \cap F_0 \not \in \mathbb{H}_0$, the value $x_1(Q)$ can be $q^2 - q$ different values, $q$ of those in $S_1$.
On the other hand,
$$| \{ x_1(Q): \: Q \in \mathcal{B} \text{ and }  Q \cap F_0 \in \mathbb{H}_0 \}| = q \left( \frac{q}{2}-1 \right),$$
since the elements of $S_0$ partitions naturally in $q-1$ subsets (\text{colors}) of size $q$, so that the places $P$, such that $x_0(P) \in S_i$, split in the same way in $q$ new places.
By having chosen $\frac{q}{2}-1$ colors, we have the previos equality.
 
Since $q \geq 8$, we have:
$$q^2 -q -q - q \left( \frac{q}{2}-1 \right) = \frac{q^2}{2}-q > q.$$
 Hence there exist some proper subset $H_1 \subset S_0$ with $q$ elements such that
$$H_1 \subset \{ x_1(Q): \: Q \cap F_1 \not \in \mathbb{H}_0, \text{ and } x_1(Q) \not \in S_1 \}$$
which allow us to construct a function
$$h_1 = \prod_{\alpha \in H_1} (x_1 - \alpha) \in V$$
of degree $q-1$ that has $q^2(q-1)$ zeros in $\mathcal{B}$.

Finally, as no zero $Q$ of $h_0$ or $h_1$ verifies $x_1(Q) \in S_1$, we can find
$$H_2 \subset \{x_2(Q): \: x_1(Q) \in S_1 \} = B_1$$
with $q-2$ elements such that the element
$$h_2 = \prod_{\alpha \in H_2} (x_2 - \alpha) \in V$$
has exactly $q^2(q-2)$ zeros in $\mathcal{B}$.

It is clear that the zeros of $h_0, h_1, h_2$ are distinct from each other,  that is, the sets  $\{Q \in \mathcal{B}: \: x_0(Q) \in H_0 \}$, $\{Q \in \mathcal{B}: \: x_1(Q) \in H_1 \}$ and $\{Q \in \mathcal{B}: \: x_2(Q) \in H_2 \}$ are disjoint by definition of $H_0, H_1$ and $H_2$.
Finally, the element
$$h = h_0h_1h_2 \in V$$
has exactly 
$$\left(\frac{q^2}{2}-q \right) \cdot q^2 + (q-1)q^2 + (q-2)q^2
$$zeros in $\mathcal{B}$,  showing that the lower bound described for $d$ is effectively attained.

\textcolor{blue}{Finally, it is possible to show that the locality of $C$ is indeed $r=q-1$ as in the proof of Theorem \ref{teoremareversionado}.}
\end{proof}

Refining the argument a little, we can show that $R = \frac{k}{n} > \frac{1}{2}$ is actually attainable with equality in the lower bound for $d$.

\begin{thm} \label{ej313} 

Let $q > 2$ be an even prime power, and let $S = \{ P_{\alpha} \in \P(F_0): \: \alpha \in S_0\}$ with $S_0 = \F_{q^2} \setminus \F_q$.
Suppose $q = 2^{2l+1}$, for
$$\mathcal{B} =\{Q \in \P(F_2): \: Q \mid P_{\alpha} \text{ for some } P_{\alpha} \in S\}$$
and $V$ the vector space spanned by
$$\left\{x_0^i x_1^j x_2^k: 0 \leq i \leq \frac{q^2}{2}, \: 0 \leq j \leq q-1, \: 0 \leq k \leq q-2\right\},$$
the code \textcolor{blue}{$C_q = C_2(\mathcal{B}, V)$} is an $\textcolor{blue}{[n,k,d]}$  LRC code \textcolor{blue}{of locality $r = q-1$ with parameters:}  
\begin{align*} n &=q^2(q^2-q)  \\
k & = \left( \frac{q^2}{2} +1 \right) q(q-1) \\
d & = n - \left(\frac{q^2}{2} \right) \cdot q^2 - (q-1)q^2 - (q-2)q^2  \\ 
& = q^2\left(  \frac{q^2}{2} - 3q +3 \right).
\end{align*}

\end{thm}

This extra condition, $q = 2^{2l+1}$, in Theorem \ref{ej313} is motivated by the next proposition:

\begin{pro}
Let $q = 2^{2l+1}$.
For all $\alpha \in S_0$, with $\alpha \in S_i$, there exist exactly two elements $\gamma \in S_i$ such that $\gamma^q + \gamma = \dfrac{\alpha^q}{\alpha^{q-1}+1}$; that is, $S_i \cap B_i = 2$ for all $1 \leq i \leq q-1$.
\end{pro}

\begin{proof}
Since $2l+1$ is odd, we have $q+1 \equiv 0 \pmod{3}$.
Let $t = \frac{q+1}{3}$ and let $\beta \in \F_{q^2}$ be primitive.
Consider the elements of the form $\alpha = \beta^{i}$ with $i \equiv t \pmod{q+1}$ such that $i \not \equiv 0 \pmod{3}$, that is, $\alpha \not \in \F_q$.
Note that $q+1$ and $q-1$ are coprime and by the Chinese Remainder Theorem, $\alpha = \beta^i$ are well defined.

Now, by considering the exponent modulo $q^2-1$, we have:
$$i \equiv t \pmod{q+1} \iff (q-1)i \equiv (q-1)t \pmod{q^2-1} \iff i + (q-1)t \equiv qi \pmod{q^2-1},$$
so that for $\alpha = \beta^i$ it is $\overline{\alpha} = \alpha^q = \alpha \cdot \beta^{(q-1)t}$.
In a similar way, since $3i \equiv 3t \equiv 0 \pmod{q+1}$, it follows that $\alpha  \cdot \beta^{2(q-1)t} \in \F_q$ and $\alpha \cdot \beta^{3(q-1)t} = \alpha$.

In particular, $( \beta^{kt})^3 = 1$ and $\beta^{2kt} = \beta^{kt}+1$,
so that
\begin{align*} \left(\alpha^q+\alpha \right)^2 & = \left( \beta^{2qi} + \beta^{2i} \right) \\
& = \beta^{i} \cdot \beta^{i} \cdot \left( \beta^{2(q-1) i} + 1 \right) \\
& = \beta^{i} \cdot \beta^{i} \cdot \left( \beta^{(q-1)i} \right)\\
& = \alpha \cdot \alpha \cdot \left( \beta^{(q-1)t} \right) \\
& =  \alpha \cdot \overline{\alpha}= \alpha^{q+1}.
\end{align*}

Since the equation $\alpha^{2q} + \alpha^2 = \alpha^{q+1}$ has at most $2q$ solutions (and one of those is $\alpha=0$ with multiplicity $2$), we have that
$$\alpha^q + \alpha = \frac{\alpha^{q}}{\alpha^{q-1}+1}$$
if and only if $\alpha = \beta^{i}$ with $i \equiv t \pmod{q+1}$ or $i \equiv 2t \pmod{q+1}$, which account $2q-2$ solutions.

In conclusion, for each of the $q-1$ subsets $S_i \subset S_0$, there exists exactly two elements $\gamma, \overline{\gamma} \in S_i$ such that $\gamma^q + \gamma = \dfrac{\gamma^q}{\gamma^{q-1}+1}$.
\end{proof}

By identifying each $S_i$ with a color, we can characterize the behaviour of each $P \in \P(F_0)$ above some place in $S$, as follows: each place of color $S_i$ has $q$ places above, of $\frac{q}{2}$ different colors, with exactly two places of each color, and in particular exactly two places of color $S_i$.

 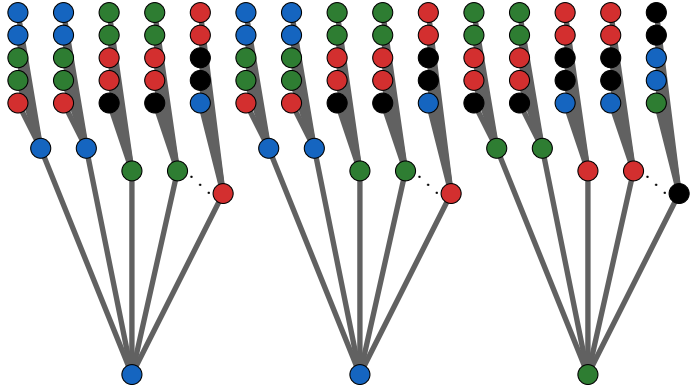
\begin{figure}[h!t]
	\begin{center}

		\begin{tikzpicture}[scale = 1.5]
			\clip(-1.4,-0.2) rectangle (5,3.4);
			\draw [line width=2pt,color=wrwrwr] (-0.8,2)-- (0,0);
			\draw [line width=2pt,color=wrwrwr] (-0.4,2)-- (0,0);
			\draw [line width=2pt,color=wrwrwr] (0,1.8)-- (0,0);
			\draw [line width=2pt,color=wrwrwr] (0.4,1.8)-- (0,0);
			\draw [line width=2pt,color=wrwrwr] (0.8,1.6)-- (0,0);
			\draw [line width=2pt,color=wrwrwr] (1.2,2)-- (2,0);
			\draw [line width=2pt,color=wrwrwr] (1.6,2)-- (2,0);
			\draw [line width=2pt,color=wrwrwr] (2,1.8)-- (2,0);
			\draw [line width=2pt,color=wrwrwr] (2.4,1.8)-- (2,0);
			\draw [line width=2pt,color=wrwrwr] (2.8,1.6)-- (2,0);
			\draw [line width=2pt,color=wrwrwr] (3.2,2)-- (4,0);
			\draw [line width=2pt,color=wrwrwr] (3.6,2)-- (4,0);
			\draw [line width=2pt,color=wrwrwr] (4,1.8)-- (4,0);
			\draw [line width=2pt,color=wrwrwr] (4.4,1.8)-- (4,0);
			\draw [line width=2pt,color=wrwrwr] (4.8,1.6)-- (4,0);
			\draw [line width=2pt,color=wrwrwr] (-1,2.4)-- (-0.8,2);
			\draw [line width=2pt,color=wrwrwr] (-1,2.6)-- (-0.8,2);
			\draw [line width=2pt,color=wrwrwr] (-1,2.8)-- (-0.8,2);
			\draw [line width=2pt,color=wrwrwr] (-1,3)-- (-0.8,2);
			\draw [line width=2pt,color=wrwrwr] (-1,3.2)-- (-0.8,2);
			\draw [line width=2pt,color=wrwrwr] (-0.6,2.4)-- (-0.4,2);
			\draw [line width=2pt,color=wrwrwr] (-0.6,2.6)-- (-0.4,2);
			\draw [line width=2pt,color=wrwrwr] (-0.6,3)-- (-0.4,2);
			\draw [line width=2pt,color=wrwrwr] (-0.6,3.2)-- (-0.4,2);
			\draw [line width=2pt,color=wrwrwr] (-0.6,2.8)-- (-0.4,2);
			\draw [line width=2pt,color=wrwrwr] (-0.2,2.4)-- (0,1.8);
			\draw [line width=2pt,color=wrwrwr] (-0.2,2.6)-- (0,1.8);
			\draw [line width=2pt,color=wrwrwr] (-0.2,2.8)-- (0,1.8);
			\draw [line width=2pt,color=wrwrwr] (-0.2,3)-- (0,1.8);
			\draw [line width=2pt,color=wrwrwr] (-0.2,3.2)-- (0,1.8);
			\draw [line width=2pt,color=wrwrwr] (0.2,2.4)-- (0.4,1.8);
			\draw [line width=2pt,color=wrwrwr] (0.2,2.6)-- (0.4,1.8);
			\draw [line width=2pt,color=wrwrwr] (0.2,2.8)-- (0.4,1.8);
			\draw [line width=2pt,color=wrwrwr] (0.2,3)-- (0.4,1.8);
			\draw [line width=2pt,color=wrwrwr] (0.2,3.2)-- (0.4,1.8);
			\draw [line width=2pt,color=wrwrwr] (0.6,2.4)-- (0.8,1.6);
			\draw [line width=2pt,color=wrwrwr] (0.6,2.6)-- (0.8,1.6);
			\draw [line width=2pt,color=wrwrwr] (0.6,3)-- (0.8,1.6);
			\draw [line width=2pt,color=wrwrwr] (0.6,2.8)-- (0.8,1.6);
			\draw [line width=2pt,color=wrwrwr] (0.6,3.2)-- (0.8,1.6);
			\draw [line width=2pt,color=wrwrwr] (1,2.4)-- (1.2,2);
			\draw [line width=2pt,color=wrwrwr] (1,2.6)-- (1.2,2);
			\draw [line width=2pt,color=wrwrwr] (1,2.8)-- (1.2,2);
			\draw [line width=2pt,color=wrwrwr] (1,3)-- (1.2,2);
			\draw [line width=2pt,color=wrwrwr] (1,3.2)-- (1.2,2);
			\draw [line width=2pt,color=wrwrwr] (1.4,2.4)-- (1.6,2);
			\draw [line width=2pt,color=wrwrwr] (1.4,2.6)-- (1.6,2);
			\draw [line width=2pt,color=wrwrwr] (1.4,3)-- (1.6,2);
			\draw [line width=2pt,color=wrwrwr] (1.4,3.2)-- (1.6,2);
			\draw [line width=2pt,color=wrwrwr] (1.4,2.8)-- (1.6,2);
			\draw [line width=2pt,color=wrwrwr] (1.8,2.4)-- (2,1.8);
			\draw [line width=2pt,color=wrwrwr] (1.8,2.6)-- (2,1.8);
			\draw [line width=2pt,color=wrwrwr] (1.8,2.8)-- (2,1.8);
			\draw [line width=2pt,color=wrwrwr] (1.8,3)-- (2,1.8);
			\draw [line width=2pt,color=wrwrwr] (1.8,3.2)-- (2,1.8);
			\draw [line width=2pt,color=wrwrwr] (2.2,2.4)-- (2.4,1.8);
			\draw [line width=2pt,color=wrwrwr] (2.2,2.6)-- (2.4,1.8);
			\draw [line width=2pt,color=wrwrwr] (2.2,2.8)-- (2.4,1.8);
			\draw [line width=2pt,color=wrwrwr] (2.2,3)-- (2.4,1.8);
			\draw [line width=2pt,color=wrwrwr] (2.2,3.2)-- (2.4,1.8);
			\draw [line width=2pt,color=wrwrwr] (2.6,2.4)-- (2.8,1.6);
			\draw [line width=2pt,color=wrwrwr] (2.6,2.6)-- (2.8,1.6);
			\draw [line width=2pt,color=wrwrwr] (2.6,3)-- (2.8,1.6);
			\draw [line width=2pt,color=wrwrwr] (2.6,2.8)-- (2.8,1.6);
			\draw [line width=2pt,color=wrwrwr] (2.6,3.2)-- (2.8,1.6);
			\draw [line width=2pt,color=wrwrwr] (3,2.4)-- (3.2,2);
			\draw [line width=2pt,color=wrwrwr] (3,2.6)-- (3.2,2);
			\draw [line width=2pt,color=wrwrwr] (3,2.8)-- (3.2,2);
			\draw [line width=2pt,color=wrwrwr] (3,3)-- (3.2,2);
			\draw [line width=2pt,color=wrwrwr] (3,3.2)-- (3.2,2);
			\draw [line width=2pt,color=wrwrwr] (3.4,2.4)-- (3.6,2);
			\draw [line width=2pt,color=wrwrwr] (3.4,2.6)-- (3.6,2);
			\draw [line width=2pt,color=wrwrwr] (3.4,2.8)-- (3.6,2);
			\draw [line width=2pt,color=wrwrwr] (3.4,3)-- (3.6,2);
			\draw [line width=2pt,color=wrwrwr] (3.4,3.2)-- (3.6,2);
			\draw [line width=2pt,color=wrwrwr] (3.8,2.4)-- (4,1.8);
			\draw [line width=2pt,color=wrwrwr] (3.8,2.6)-- (4,1.8);
			\draw [line width=2pt,color=wrwrwr] (3.8,2.8)-- (4,1.8);
			\draw [line width=2pt,color=wrwrwr] (3.8,3)-- (4,1.8);
			\draw [line width=2pt,color=wrwrwr] (3.8,3.2)-- (4,1.8);
			\draw [line width=2pt,color=wrwrwr] (4.2,2.4)-- (4.4,1.8);
			\draw [line width=2pt,color=wrwrwr] (4.2,2.6)-- (4.4,1.8);
			\draw [line width=2pt,color=wrwrwr] (4.2,2.8)-- (4.4,1.8);
			\draw [line width=2pt,color=wrwrwr] (4.2,3)-- (4.4,1.8);
			\draw [line width=2pt,color=wrwrwr] (4.2,3.2)-- (4.4,1.8);
			\draw [line width=2pt,color=wrwrwr] (4.8,1.6)-- (4.6,2.4);
			\draw [line width=2pt,color=wrwrwr] (4.6,2.6)-- (4.8,1.6);
			\draw [line width=2pt,color=wrwrwr] (4.8,1.6)-- (4.6,2.8);
			\draw [line width=2pt,color=wrwrwr] (4.6,3)-- (4.8,1.6);
			\draw [line width=2pt,color=wrwrwr] (4.8,1.6)-- (4.6,3.2);
			\begin{scriptsize}
				\draw [fill=rvwvcq] (2,0) circle (2.5pt);
				\draw [fill=sexdts] (4,0) circle (2.5pt);
				\draw [fill=rvwvcq] (0,0) circle (2.5pt);
				\draw [fill=rvwvcq] (-0.8,2) circle (2.5pt);
				\draw [fill=rvwvcq] (-0.4,2) circle (2.5pt);
				\draw [fill=sexdts] (0,1.8) circle (2.5pt);
				\draw [fill=sexdts] (0.4,1.8) circle (2.5pt);
				\draw [fill=dtsfsf] (0.8,1.6) circle (2.5pt);
				\draw [fill=sexdts] (2,1.8) circle (2.5pt);
				\draw [fill=rvwvcq] (1.6,2) circle (2.5pt);
				\draw [fill=rvwvcq] (1.2,2) circle (2.5pt);
				\draw [fill=sexdts] (2.4,1.8) circle (2.5pt);
				\draw [fill=dtsfsf] (2.8,1.6) circle (2.5pt);
				\draw [fill=sexdts] (3.2,2) circle (2.5pt);
				\draw [fill=sexdts] (3.6,2) circle (2.5pt);
				\draw [fill=dtsfsf] (4,1.8) circle (2.5pt);
				\draw [fill=dtsfsf] (4.4,1.8) circle (2.5pt);
				\draw [fill=black] (4.8,1.6) circle (2.5pt);
				\draw [fill=sexdts] (-1,2.6) circle (2.5pt);
				\draw [fill=sexdts] (-1,2.8) circle (2.5pt);
				\draw [fill=rvwvcq] (-1,3) circle (2.5pt);
				\draw [fill=rvwvcq] (-1,3.2) circle (2.5pt);
				\draw [fill=dtsfsf] (-1,2.4) circle (2.5pt);
				\draw [fill=dtsfsf] (-0.6,2.4) circle (2.5pt);
				\draw [fill=sexdts] (-0.6,2.6) circle (2.5pt);
				\draw [fill=sexdts] (-0.6,2.8) circle (2.5pt);
				\draw [fill=rvwvcq] (-0.6,3) circle (2.5pt);
				\draw [fill=rvwvcq] (-0.6,3.2) circle (2.5pt);
				\draw [fill=black] (-0.2,2.4) circle (2.5pt);
				\draw [fill=dtsfsf] (-0.2,2.6) circle (2.5pt);
				\draw [fill=dtsfsf] (-0.2,2.8) circle (2.5pt);
				\draw [fill=sexdts] (-0.2,3) circle (2.5pt);
				\draw [fill=sexdts] (-0.2,3.2) circle (2.5pt);
				\draw [fill=black] (0.2,2.4) circle (2.5pt);
				\draw [fill=dtsfsf] (0.2,2.6) circle (2.5pt);
				\draw [fill=dtsfsf] (0.2,2.8) circle (2.5pt);
				\draw [fill=sexdts] (0.2,3) circle (2.5pt);
				\draw [fill=sexdts] (0.2,3.2) circle (2.5pt);
				\draw [fill=rvwvcq] (0.6,2.4) circle (2.5pt);
				\draw [fill=black] (0.6,2.6) circle (2.5pt);
				\draw [fill=black] (0.6,2.8) circle (2.5pt);
				\draw [fill=dtsfsf] (0.6,3) circle (2.5pt);
				\draw [fill=dtsfsf] (0.6,3.2) circle (2.5pt);
				\draw [fill=dtsfsf] (1,2.4) circle (2.5pt);
				\draw [fill=sexdts] (1,2.6) circle (2.5pt);
				\draw [fill=sexdts] (1,2.8) circle (2.5pt);
				\draw [fill=rvwvcq] (1,3) circle (2.5pt);
				\draw [fill=rvwvcq] (1,3.2) circle (2.5pt);
				\draw [fill=dtsfsf] (1.4,2.4) circle (2.5pt);
				\draw [fill=sexdts] (1.4,2.6) circle (2.5pt);
				\draw [fill=rvwvcq] (1.4,3) circle (2.5pt);
				\draw [fill=rvwvcq] (1.4,3.2) circle (2.5pt);
				\draw [fill=sexdts] (1.4,2.8) circle (2.5pt);
				\draw [fill=black] (1.8,2.4) circle (2.5pt);
				\draw [fill=dtsfsf] (1.8,2.6) circle (2.5pt);
				\draw [fill=dtsfsf] (1.8,2.8) circle (2.5pt);
				\draw [fill=sexdts] (1.8,3) circle (2.5pt);
				\draw [fill=sexdts] (1.8,3.2) circle (2.5pt);
				\draw [fill=black] (2.2,2.4) circle (2.5pt);
				\draw [fill=dtsfsf] (2.2,2.6) circle (2.5pt);
				\draw [fill=dtsfsf] (2.2,2.8) circle (2.5pt);
				\draw [fill=sexdts] (2.2,3) circle (2.5pt);
				\draw [fill=sexdts] (2.2,3.2) circle (2.5pt);
				\draw [fill=rvwvcq] (2.6,2.4) circle (2.5pt);
				\draw [fill=black] (2.6,2.6) circle (2.5pt);
				\draw [fill=dtsfsf] (2.6,3) circle (2.5pt);
				\draw [fill=black] (2.6,2.8) circle (2.5pt);
				\draw [fill=dtsfsf] (2.6,3.2) circle (2.5pt);
				\draw [fill=black] (3,2.4) circle (2.5pt); 
				\draw [fill=dtsfsf] (3,2.6) circle (2.5pt);
				\draw [fill=dtsfsf] (3,2.8) circle (2.5pt);
				\draw [fill=sexdts] (3,3) circle (2.5pt);
				\draw [fill=sexdts] (3,3.2) circle (2.5pt);
				\draw [fill=black] (3.4,2.4) circle (2.5pt);
				\draw [fill=dtsfsf] (3.4,2.6) circle (2.5pt);
				\draw [fill=dtsfsf] (3.4,2.8) circle (2.5pt);
				\draw [fill=sexdts] (3.4,3) circle (2.5pt);
				\draw [fill=sexdts] (3.4,3.2) circle (2.5pt);
				\draw [fill=rvwvcq] (3.8,2.4) circle (2.5pt);
				\draw [fill=black] (3.8,2.6) circle (2.5pt);
				\draw [fill=black] (3.8,2.8) circle (2.5pt);
				\draw [fill=dtsfsf] (3.8,3) circle (2.5pt);
				\draw [fill=dtsfsf] (3.8,3.2) circle (2.5pt);
				\draw [fill=rvwvcq] (4.2,2.4) circle (2.5pt);
				\draw [fill=black] (4.2,2.6) circle (2.5pt);
				\draw [fill=black] (4.2,2.8) circle (2.5pt);
				\draw [fill=dtsfsf] (4.2,3) circle (2.5pt);
				\draw [fill=dtsfsf] (4.2,3.2) circle (2.5pt);
				\draw [fill=sexdts] (4.6,2.4) circle (2.5pt);
				\draw [fill=rvwvcq] (4.6,2.6) circle (2.5pt);
				\draw [fill=rvwvcq] (4.6,2.8) circle (2.5pt);
				\draw [fill=black] (4.6,3) circle (2.5pt);
				\draw [fill=black] (4.6,3.2) circle (2.5pt);
			\end{scriptsize}
			\node  at(4.6,1.75){\footnotesize{$\ddots$}};
			\node  at(2.6,1.75){\footnotesize{$\ddots$}};
			\node  at(0.6,1.75){\footnotesize{$\ddots$}};
		\end{tikzpicture}
		\caption{Diagram of three splitting places $P_j$ of $F_0$ in $F_2/F_0$, for $q=2^{2l} > 5$.
Each color, in each function field, represent a set $S_i$, for $1 \leq i \leq 4=q-1$, so each place has exactly two places of its same color above it. }\label{figuplaces}
\end{center}\end{figure}

We now proceed to prove the Theorem \ref{ej313}

\begin{proof}

Similarly as in the proof or the previous theorem, take $S_1 \subset S_0$ arbitrarily, with $q$ elements.
We can see that $| \{ P \in \P(F_1): \: x_1(P) \in S_1 \} | = q^2$ since for each $\beta \in S_1$ there are exactly $q$ places $P_{\alpha} \in \P(F_0)$ lying under some place $P$ with $x_1(P) = \beta$.
The last part of Lemma \ref{lema3} asserts those $q^2$ places come \textit{in pairs}: by thinking each set $S_i$ as one of $q-1$ \textit{colors},  we see that for each $P_{\alpha} \in \P(F_0)$ there exists two places $P, \: P'$ of color $S_1$ according to the value $x_1(P)$.
In particular, $P' = \sigma(P)$ where $\sigma$ is the only non trivial automorphism of $\F_{q^2}/\F_q$.

This pairing shows that, in consequence,  there exists exactly $\frac{q^2}{2}$ places $P_{\alpha} \in S$ with exactly $2$ places above each of them of color $S_1$.

Consider the remaining $q^2-q - \frac{q^2}{2}$ places, that is,
$$\mathbb{H}_0 = \{ P_{\alpha} \in \P(F_0) : \: x_1(P) \not\in S_1 \: \text{ for all } P|P_{\alpha}\}$$
of size $\frac{q^2}{2} - q$ and let
$$H_0  = \{\alpha \in S_0 : \: P_{\alpha} \in \mathbb{H}_0 \} \cup S_1$$
of size $\frac{q^2}{2} -q+q =\dfrac{q^2}{2}$.

The element
$$h_0 =  \prod_{\alpha \in H_0} (x_0 - \alpha) \in V$$
has exactly $q^2 \cdot \frac{q^2}{2} $ zeros in $\mathcal{B}$, precisely those places $Q$ lying above some place in $\mathbb{H}_0$ and those places $Q$ with $x_0(Q) \in S_1$.

Among the remaining places $Q \in \mathcal{B}$ such that $x_0(Q) \not\in H_0$, the value $x_1(Q)$ can attain $q^2 - q$ values, $q$ of those are in $S_1$.
On the other hand,
$$| \{ x_1(Q): \: Q \in \mathcal{B} \text{ and }   x_0(Q) \not\in H_0 \}| = q \cdot  \frac{q}{2}$$
since the elements of $S_0$ partitions naturally in $q-1$ subsets (\text{colors}) of size $q$, so that the places $P$ with $x_0(P) \in S_i$ split in the same way in $q$ new places.
By having chosen $\frac{q}{2}$ colors, we have the previous equality.
 
Since $q \geq 8$, we have
$$\left(q^2 -q \right)-q - q \left( \frac{q}{2} \right) = \frac{q^2}{2}-2q > q.$$
Hence there exist some proper subset $H_1 \subset S_0$ with $q$ elements such that
$$H_1 \subset \{ x_1(Q): \: Q \cap F_0 \not \in \mathbb{H}_0, \text{ and } x_1(Q) \not \in S_1 \}$$
which allow us to construct an element
$$h_1 = \prod_{\alpha \in H_1} (x_1 - \alpha) \in V$$
of degree $q-1$ that has $q^2(q-1)$ zeros in $\mathcal{B}$.

Finally, as no zero $Q$ of $h_0$ or $h_1$ verifies that $x_1(Q) \in S_1 \setminus B_1$, we have a non-empty subset
$$H_2 \subset \{x_2(Q): \: x_1(Q) \in S_1 \text{ and} x_0(Q) \not \in H_0 \} = S_1 \setminus B_1$$
with $q-2$ elements such that the element
$$h_2 = \prod_{\alpha \in H_2} (x_2 - \alpha) \in V$$
has exactly $q^2(q-2)$ zeros in $\mathcal{B}$.

It is clear that the zeros of $h_0, h_1$ and $h_2$ are distinct from each other,  that is, the sets  $\{Q \in \mathcal{B}: \: x_0(Q) \in H_0 \}$ and $\{Q \in \mathcal{B}: \: x_1(Q) \in H_1 \}$, $\{Q \in \mathcal{B}: \: x_2(Q) \in H_2 \}$ are disjoint by definition of $H_0, H_1$ and $H_2$.
Finally, the element
$$h = h_0h_1h_2 \in V$$
has exactly 
$$\frac{q^2}{2}  + (q-1) q^2 +  (q-2)q^2
$$zeros in $\mathcal{B}$,  showing that the lower bound described for $d$ is effectively attained.

\textcolor{blue}{The locality of $C$ is also $r=q-1$ as in the proof of Theorem \ref{teoremareversionado}.}
\end{proof}

In conclusion, we have shown two examples of LRC codes of length $n = q^4-\textcolor{blue}{q^3}$ and relative parameters $(\delta, R) = \left( \dfrac{1}{2} - \frac{3}{2q} + o(q^{-2}), \: \dfrac{1}{2} - \frac{1}{q} + \frac{1}{q^2} \right)$ and $(\delta, R) = \left( \dfrac{1}{2} - \frac{5}{2q} + o(q^{-2}), \:\dfrac{1}{2} + \frac{1}{q^2} \right)$ respectively.
It is worth to notice that, due to the good recursive behavior of the splitting of the places in the tower, the same results can be obtained for the extensions $F_0 \subset F_i \subset F_{i+1}$ in the tower, for any $i \geq 1$.

For the case of two steps, these two constructions have known dimension, distance, and were obtained as an evaluation code.

\begin{cor}\label{coro38}
Let $q = 2^{2l+1}$.
For $n=q^2(q^2-q)$, $1 \leq l \leq \frac{q^2}{2}$,  there exists an $\textcolor{blue}{[n,k,d]}$ LRC Evaluation Code with locality $r=q-1$ and known parameters $k = (l+1)(q^2-q)$ and $d = q^2 \left( q^2 - l - 3q+3 \right)$ , where each element $f \in V \subset \F_{q^2}(x_0,x_1,x_2)$ is a polynomial in $x_0, x_1, x_2$.
\end{cor}

\begin{proof}
\textcolor{blue}{As we have shown in Theorems \ref{ej313} and \ref{ej312}, we obtain an LRC code with known dimension and minimum distance. By slightly reducing the dimension of $V$, we obtain an LRC code with length $n= q^4-q^3$, dimension $k =  (l+1)(q^2-q)$, since the vector space $V$ is spanned by $$\left\{x_0^{i}x_1^{j}x_2^{k}: \: 0 \leq i \leq l, \: 0 \leq j \leq q-1, \: 0 \leq k \leq q-2\right\},$$ and minimum distance $d \geq n - l\deg(x_0) - (q-1)\deg(x_1) - (q-2)\deg(x_2) = q^2(q^2-l-3q+3)$. But in this cases we also know that there exist a function $h \in V$ that is a polynomial in $x_0, x_1, x_2$ with such many zeros in $\mathcal{B}$, so $d  = q^2(q^2-l-3q+3)$.
}\end{proof}

\section{Other Constructions }\label{segunda}\

In this section we will continue using towers of function fields in order to build more sequences of LRC codes

\subsection{The Garcia - Stichtenoth tower over $\F_4$}

Consider the tower of function fields
$$\mathcal{T} = (F_0, F_1, F_2, \ldots),$$
given by $F_0 = \F_4(x_0)$, and $F_i = F_{i-1}(x_i)$ with equation
$$x_i ^2 + x_i = \dfrac{x_{i-1}^2}{x_{i-1} + 1}.$$

\begin{pro}\label{ej1} There is a LRC AG Code $C$ in the function field $F_j/F_0$, with length $n=2^j$, dimension $k=2^{j-2} = \frac{1}{4}n$, minimum distance $d=2$ and locality $r=1$.
Moreover, such $C$ is equivalent to a repetition code.
\end{pro}

\begin{proof}
This proposition is a direct consequence of the application of the construction of Theorem \ref{teoremareversionado}, and so $r = m-1 = 2-1 = 1$ and $d = 2$.
Indeed, we can characterize the rational places of $\mathcal{T}$ as follows.
The rational places of $F_0$ are precisely $\P(F_0) = \{P_0, P_1, P_{\alpha}, P_{\alpha+1}, P_{\infty} \}$, whereas on \cite{gusti} it is proved that the only rational places of $F_0$ that splits completely on $\mathcal{T}$ are $P_{\alpha}$ and $P_{\alpha+1}$ for $\alpha^2+\alpha+1=0$.
Moreover, for each $i \geq 0$, the place $P_{\infty} \in F_i$ is totally ramified, so there exists only one place $Q_{\infty} \in F_{i+1}$ above  $P_{\infty}$.
Noticing that $[F_{i+1}:F_i] = 2$, we can construct a linear code of length $2 \cdot 2^j$ and dimension $k = 2 \cdot 2^{j-2}$ considering the set of places $\mathcal{B} = \{ Q \mid P_{\alpha} \} \cup \{ Q \mid P_{\alpha^2} \} $ for $P_{\alpha+1} = P_{\alpha^2}$ in $F_0$, where each one splits in $2^j$ total places in $F_j$.

Since $V \subset \mathcal{L}(Q_{\infty})$, for $D = 2^{j+1}Q_{\infty}$ we can consider the code \textcolor{blue}{$C = C_2(\mathcal{B},V)$} as an evaluation code of places of $F_{i+1}$; that will have the same parameters shown, e.g., $C$ is \textcolor{blue}{ an $[2^j, 2^{j-1}, 2^{j-1}]$} linear code with $r=m-1=1$.
\end{proof}

Furthermore, for more examples of codes in this function field, we may think in subtle variations, as we show next.

 Consider some subtower $\mathcal{E} = (E_0, E_1, E_2, \ldots)$ of the tower $\mathcal{T}$ in which $E_0 = F_0 = \F_4(x_0)$ and $E_j \subset F_{2j}$.

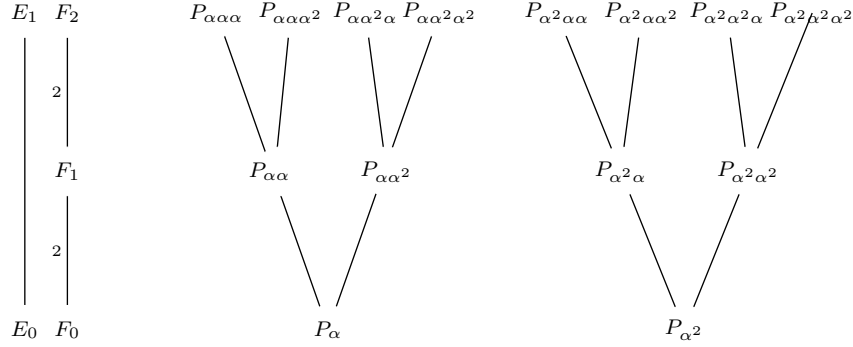
\begin{figure}[h!t]
\begin{center}
\begin{tikzpicture}[scale=1.4]

\draw[line width=0.5 pt](-0.9,0)--(-0.9,3);
\draw[line width=0.5 pt](-1.3,0)--(-1.3,3);
\draw[white, fill=white](-0.9,0) circle (0.23 cm);
\draw[white, fill=white](-0.9,1.5) circle (0.23 cm);
\draw[white, fill=white](-0.9, 3) circle (0.23 cm);
\draw[white, fill=white](-1.3, 0) circle (0.23 cm);
\draw[white, fill=white](-1.3, 3) circle (0.23 cm);
\node at(-0.9,0){\footnotesize{$F_0$}};
\node at(-0.9,1.5){\footnotesize{$F_1$}};
\node at(-1.3,0){\footnotesize{$E_0$}};
\node at(-1.3,3){\footnotesize{$E_1$}};
\node at(-0.9,3){\footnotesize{$F_2$}};
\node at(-1, 2.25){\tiny{$2$}};
\node at(-1, 0.75){\tiny{$2$}};

\draw[line width=0.5 pt](1.55,0)--(0.5,3);
\draw[line width=0.5 pt](1.55,0)--(2.6,3);
\draw[line width=0.5 pt](4.9,0)--(3.7,3);
\draw[line width=0.5 pt](4.9,0)--(6.1,3);

\draw[line width=0.5 pt](1.05,1.5)--(1.2,3);
\draw[line width=0.5 pt](2.1,1.5)--(1.9,3);
\draw[line width=0.5 pt](4.3,1.5)--(4.5,3);
\draw[line width=0.5 pt](5.5,1.5)--(5.3,3);

\draw[white, fill=white](0.5,3) circle (0.23 cm);
\draw[white, fill=white](1.2,3) circle (0.23 cm);
\draw[white, fill=white](1.9,3) circle (0.23 cm);
\draw[white, fill=white](2.6,3) circle (0.23 cm);
\draw[white, fill=white](3.7,3) circle (0.23 cm);
\draw[white, fill=white](4.5,3) circle (0.23 cm);
\draw[white, fill=white](5.3,3) circle (0.23 cm);
\draw[white, fill=white](6,1,3) circle (0.23 cm);
\node at(0.5,3){\footnotesize{$P_{\alpha \alpha \alpha}$}};
\node at(1.2,3){\footnotesize{$P_{\alpha \alpha  \alpha^2 }$}};
\node at(1.9,3){\footnotesize{$P_{\alpha  \alpha^2  \alpha}$}};
\node at(2.6,3){\footnotesize{$P_{\alpha  \alpha^2   \alpha^2 }$}};
\node at(3.7,3){\footnotesize{$P_{ \alpha^2  \alpha \alpha}$}};
\node at(4.5,3){\footnotesize{$P_{ \alpha^2  \alpha  \alpha^2 }$}};
\node at(5.3,3){\footnotesize{$P_{ \alpha^2   \alpha^2  \alpha}$}};
\node at(6.1,3){\footnotesize{$P_{ \alpha^2   \alpha^2   \alpha^2 }$}};

\draw[white, fill=white](1.1,1.5) circle (0.23 cm); 
\draw[white, fill=white](2.1,1.5) circle (0.23 cm);
\draw[white, fill=white](4.3,1.5) circle (0.23 cm);
\draw[white, fill=white](5.5,1.5) circle (0.23 cm);
\node at(1,1.5){\footnotesize{$P_{\alpha \alpha}$}};
\node at(2.1,1.5){\footnotesize{$P_{\alpha  \alpha^2 }$}};
\node at(4.3,1.5){\footnotesize{$P_{ \alpha^2  \alpha}$}};
\node at(5.5,1.5){\footnotesize{$P_{ \alpha^2   \alpha^2 }$}};

\draw[white, fill=white](1.55,0) circle (0.23 cm);
\draw[white, fill=white](4.9,0) circle (0.23 cm);
\node at(1.55,0){\footnotesize{$P_{\alpha}$}};
\node at(4.9,0){\footnotesize{$P_{ \alpha^2 }$}};
\end{tikzpicture}
\caption{Diagram of the splitting of the rational places $P$ en $F_0$ that splits completely.
From left to right, we may label the rational places of $F_2$ as $P_1, P_2, \ldots, P_8$.}\label{figuf4}
\end{center}\end{figure}

\begin{rem} \phantom{ }

\begin{enumerate}   \label{ej2}

\item  \label{ej2a} Considering the extension $E_1/E_0$, we obtain an $[8, 4, 2]$  LRC Code $C$ with $r = 1$.
Indeed, for
$$V =  \left \langle 1 ,x,y,xy \right \rangle,$$
then any $f \in V$ verifies that $f(P_{1}) = f(P_2)$, $f(P_3)=f(P_4)$, $f(P_5)=f(P_6)$ and $f(P_7)=f(P_8)$, where $y \in F_1$ verifies $y^2+y = \frac{x^2}{x+1}$.

\item \label{ej2b} Considering the vector space $W = \left \langle 1 ,x,z,xz \right \rangle $ instead of $V = \left \langle 1, x, y ,xy \right \rangle$ again we have an $[8, 4, 2]$ LRC code $C$ with $r = 1$, where $z \in F_2$ verifies $z^2+z = \frac{y^2}{y+1}$.
\end{enumerate}
\end{rem}

\subsection{The van der Geer - van der Vlugt tower over $\F_8$}

Consider the tower of function fields
$$\mathcal{T} = (F_0, F_1, F_2, \ldots),$$
given by the relations $F_0 = \F_8(x_0)$ and $F_{i+1}= F_{i}(x_{i+1})$ with recursive equation
$$x_{i+1}^2+x_{i+1} = x_i + 1 + \frac{1}{x_i} .$$

\begin{pro}[\cite{garciaetal}]
The set $S$ of rational places in $F_0$ that splits completely in $\mathcal{T}$ (that is, rational places $P$ that splits completely on $F_i/F_0$ for every $i$) is characterized by $S = \{P_{\alpha}: \: \alpha \in \F_8 \setminus \F_2 \}$.
In particular if  $\F_8 = \F_2(\beta)$ with $\beta^3 = \beta+1$, we have\footnote{Here we use $\beta$ as a generator of $\F_8^*$. Later on, we will use it as an arbitrary element of $\F_8 \setminus \F_2$.}
$$S_0 = \F_8 \setminus \F_2 =  \{\beta,\: \beta+1,\: \beta^2,\: \beta^2+1, \: \beta^2+\beta,\: \beta^2+\beta+1\},$$
and consequently
$$S = \{P_{\beta}: \: \beta \in S_0\}.$$
\end{pro}

\begin{pro}  \label{ej37} Let $\mathcal{B} = \{  Q \in \P(F_2) : \: Q \mid P \text{ for some } P \in S \}$ and let $V$ be the vector space spanned by $\{x_0^i x_1 ^j : \: 0 \leq i \leq 4, \: 0 \leq j \leq 1 \}$.
Then the linear code $\C_2 = \C( \mathcal{B}, V)$ is a $\textcolor{blue}{[24, 10, 4]}$ LRC code with \textcolor{blue}{$r = 1$.}
\end{pro}

\begin{proof}
Since $[F_2:F_0] = 4$, then $|\mathcal{B}| = 4 \cdot 6 = 24$.
Moreover, it is clear that $\dim V = 10$ since its generators are linearly independent.
Finally, due to the construction of $C_2$, we have
$$d \geq 24 -4 \cdot \deg(x_0) - \deg(x_1).$$

 Since $[F_2: \F_8(x_0)] = 4 = [F_2: \F_8(x_1)] = 4$ we have
$$d \geq 24 -16 -4 = 4.$$
On the other hand, for $S_0 = \{ \beta_1, \ldots, \beta_6\}$ and $\F_8 = \F_2(\beta_1)$ with $\beta_1 ^3 =\beta_1 + 1$, the element $f = (x_1 - \beta_1)(x_0-\beta_2)(x_0-\beta_3)(x_0 - \beta_4)(x_0 - \beta_5)$ has exactly $20$ simple zeros in $F_2$.

\textcolor{blue}{Finally, $r=1$ as in proposition \ref{ej1} since $[F_i:F_{i-1}]=2$ for $i \geq 1$. In fact, for any $Q \in \mathcal{B}$, and $Q'$ the other place above $P = Q \cap F_2$, and for any $f \in V$, it is $f(Q) = f(Q')$. }
\end{proof}

\begin{pro} \label{ej38} 

Let $\mathcal{B} = \{ Q \in \P(F_3) : \: Q \mid P \text{ for some } P \in S \}$ and let $V$ be the vector space spanned by $\{x_0^i x_1 ^j x_2^k: \: 0 \leq i \leq 4, \: 0 \leq j,k \leq 1 \}$.
Then the code $C_3 = C( \mathcal{B}, V)$ is a $[48, 20, d]$ LRC code with $d \leq 4$, \textcolor{blue}{and $r = 1$}.
\end{pro}

\begin{proof}
As in the previous proposition, we have that $\dim V = 20$.
Now,
$$d \geq 48 - 4 \deg(x_0) - \deg(x_1) - \deg(x_2),$$
or equivalently,
$$d \geq 48 - 32 - 8 -8 \geq 0.$$

We are going to show that $d\geq 4$ with a similar argument as before.
In fact, the element
$$f = (x_0-\beta)(x_0-\beta-1)(x_0-\beta^2-\beta)(x_0-  \beta^2-1)(x_1- \beta^2)(x_2 - \beta)
$$has $4\cdot 2^3 + 2 \cdot 2^2 + 2 \cdot 2^1 = 32 + 8 +4 = 44$ distinct simple zeros in $F_3$, showing that $d \leq 4$.
\end{proof}

In order to justify the construction of such $f$ in the last proof, we may do the following analysis, to also suspect that $d \geq 4$.

Consider the directed graph $G$ whose vertices are the elements of $S_0 = \F_8 \setminus \F_2$ and has edges $( \alpha, \beta)$ with $\beta^2 + \beta = \alpha + 1 + \frac{1}{\alpha}$.

\begin{figure}[h!t]
\begin{center}
\begin{tikzpicture}[scale=1.4]

\draw[white, fill=white](1.8,2) circle (0.23 cm);
\node (A) at(1.8,2){\footnotesize{$\beta$}};
\draw[white, fill=white](4,2) circle (0.23 cm);
\node (B) at(4,2){\footnotesize{$\beta+1$}};
\draw[white, fill=white](4.8,1) circle (0.23 cm);
\node (C) at(4.8,1){\footnotesize{$\beta^2+\beta$}};
\draw[white, fill=white](4,0) circle (0.23 cm);
\node (D) at(4,0){\footnotesize{$ \beta^2 + \beta +1$}};
\draw[white, fill=white](1.8,0) circle (0.23 cm);
\node (E) at(1.8,0){\footnotesize{$ \beta^2 $}};
\draw[white, fill=white](1,1) circle (0.23 cm);
\node (F) at(1,1){\footnotesize{$\beta^2+1$}};

\draw [->] (A) -- (B);
\draw [->] (B) -- (C);
\draw [->] (C) -- (D);
\draw [->] (D) -- (E);
\draw [->] (E) -- (F);
\draw [->] (F) -- (A);
\draw [->] (B) -- (D);
\draw [->] (D) -- (F);
\draw [->] (F) -- (B);
\draw[->] (A) to[out=45,in=135,looseness=5] (A);
\draw[->] (C) to[out=315,in=45,looseness=5] (C);
\draw[->] (E) to[out=315,in=225,looseness=5] (E);
\end{tikzpicture}
\caption{Graph relating the elements of $S_0 = \F_8 \setminus \F_2$, with an edge from $\alpha$ to $\beta$ for each pair that verifies the recursive equation that defines the tower $\mathcal{T}$ .}\label{grafof8}
\end{center}\end{figure}
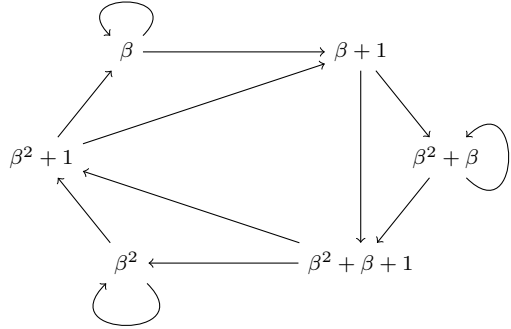

We can see that each vertex has exactly two outgoing edges, and this is because each place $P_{\alpha} \in \P(F_0)$ with $\alpha \in S_0$ splits in two places.
In particular $\beta^2+ \beta = (\beta+1)^2 + (\beta+1)$ for every $\beta \in S_0$, and analogously each vertex has exactly two ongoing edges, where $\alpha+1 + \frac{1}{\alpha} = \alpha^{-1}+1 + \frac{1}{\alpha^{-1}}$ for every $\alpha \in S_0$.

It is convenient to note the following: \begin{lem} $Q \in \mathcal{B} \subset F_3$ is a zero of $x_i - \beta$ for $0 \leq i \leq 3$ if and only if there exists $Q \mid P_{\alpha} \in S$ such that there is a path of length $i$ from $\alpha$ to $\beta$ in $G$.
\end{lem}

\begin{proof}
Since $x_{i+1}^2+x_{i+1} = x_i + 1 + \frac{1}{x_i}$, we have that $x_i(Q) = \alpha$ if and only if $x_{i+1}(Q) = \beta,$ where $(\alpha, \beta)$ is some edge in $G$.
Since $x_0(Q) = x_0(P_{\alpha}) = \alpha$, the lemma follows by induction on $i$.
\end{proof}

We notice then that each function $x_0 - \beta$ has exactly $8  = 2^3$ zeros $Q \in \mathcal{B}$ (i.e., $\deg(x_0- \beta)_0 = 8$ and the $8$ places above $P_{\beta}$ are zeros of $x_0-\beta)$.
For $i=1$, each element $x_1 - \beta$ has exactly $2 \cdot 2^2 = 8$ zeros $Q \in \mathcal{B}$ since there are exactly $2$ paths of length $1$ that ends in $\beta$.

It is also worth to notice that $x_i- \beta$ and $x_i - \beta - 1$ has the same zeros in $\mathcal{B}$, and similarly $x_i - \alpha$ and $x_{i+1} - \beta$ for $\beta^2 + \beta = \alpha + 1 + \frac{1}{\alpha}$.
Moreover, since for each pair $\alpha, \beta \in G$ there exist some path of length less or equal than $3$, we get that $x_3 - \alpha$ and $x_0 - \beta$ will have a common zero in $\mathcal{B}$.

This analysis enables us to conjecture that any element spanned by 
$$1, \: x_0, \: x_0^2,\: x_0^3,\: x_0^4$$
would have at least $4$ zeros in common with any of the elements $x_1 - \alpha$ and $x_2 - \beta$.
In particular, we can conjecture that no element in $V$ has more than  $48-4 = 44$ zeros in $\mathcal{B}$, where $V$ is the vector space spanned by
$$\{x_0^i x_1 ^j x_2^k: \: 0 \leq i \leq 4, \: 0 \leq j,k \leq 1 \}.$$

\section{Conclusions} \label{conclusiones}

In this last section, we want to consider the codes constructed and compare its relative parameters $\delta$ and $R$, together with their relative locality $\frac{r}{n}$.
Recall that $q \geq 8$ is an arbitrarily large power of $2$.
\begin{figure}[ht] \label{table}
\renewcommand{\arraystretch}{1.5}
\begin{tabular}{ |c|c|c|c|c|c| } \hline
Example  &$n$ & $\delta$ & $R$ & $\frac{r}{n}$ \\ \hline
  \ref{ej312}  & $ q^4-q^2 $ & $ \frac{1}{2}- \frac{3}{2q} + o(q^{-2}) $ & $ \frac{1}{2}-\frac{1}{q} + \frac{1}{q^2}$ & $ \frac{1}{q^3} $ \\ \hline
  \ref{ej313}  & $ q^4-q^2 $ & $ \frac{1}{2}- \frac{5}{2q} + o(q^{-2}) $ & $ \frac{1}{2} + \frac{1}{q^2}$ & $ \frac{1}{q^3} $ \\ \hline 
 \ref{ej1} & $ 2\cdot 2^i $ & $ \frac{1}{2^i} $ & $ \frac{1}{4} $ & $ \frac{1}{2^{i+1}} $ \\ \hline 
\ref{ej2}, \ref{ej2a} & $ 8 $ & $ \frac{1}{4} $ & $ \frac{1}{2} $ & $ \frac{1}{8} $ \\ \hline 
\ref{ej2}, \ref{ej2b} & $ 8 $ & $ \frac{1}{4} $ & $ \frac{1}{2} $ & $ \frac{1}{8} $ \\ \hline 
\ref{ej37} & $ 24 $ & $ \frac{1}{6} $ & $ \frac{5}{12} $ & $ \frac{1}{24} $ \\ \hline
 \ref{ej38} & $ 48 $ & $ \leq \frac{1}{12} $ & $ \frac{5}{12} $ & $ \frac{1}{48} $ \\ \hline 

 \end{tabular}
 \caption{Table with the studied examples} \label{tabla}
\end{figure}

For comparison purposes, we plotted in figure \ref{ejemplotabla} the values of $\delta$ and $R$ with different colors: cyan for Examples \ref{ej1}, blue for Examples \ref{ej2}, red for Examples \ref{ej37} and \ref{ej38}, and finally black for Examples \ref{ej312} and \ref{ej313} with $q=8$ and $q=32$.

For completeness, we also plot Examples \ref{ej312} and \ref{ej313} together with other possible examples that follow from Remark $3.6$ of \cite{lrcontowers}, and the bounds derived in \cite{bargtamobounds} and \cite{bargtamovladut} in figure \ref{ejemploq32}.
These results manifest the existence of LRC Codes with good relative parameters, for fixed $q$ and increasing length $n$, in contrast of the codes shown where $q$ is arbitrarily large and $n \approx q^4$.
In particular, there can be constructed explicit examples as in \ref{ej313} with smaller subspaces $V' \subset V$, obtaining relative distances $\delta$ arbitrarily close to $1$, and transmission rate $R$ lying above strict bounds improving a notorious bound of Barg-Tamo-Vladut.

\begin{pro}[Barg-Tamo-Vladut, \cite{bargtamovladut}, Remark 6.3]
For each $q = 2^{2l}$ there exist a family of LRC codes, of locality $r = \sqrt{q}-1$, whose relative parameters verify

\begin{equation}\label{BTVbound} R\geq \frac{r}{r+1}\left(1-\delta-\frac{3}{q+1}\right).\end{equation}
\end{pro}

The relationship of the relative parameters of the codes constructed in Theorems \ref{ej312} and \ref{ej313} rely instead on the arbitrarily large size of the base field $\F_{q^2}$, but verify a stronger inequality similar as the obtained in Remark 3.6 of \cite{lrcontowers}.
\begin{pro}
For $q \geq 8$ and for each $1 \leq l \leq (q-1)(q-2)$; the relative parameters $(R, \delta)$ of the code $C_q$, constructed in Example \ref{ej313} verify
\begin{equation*} R + \frac{q-1}{q} \delta > \frac{q-1}{q} \left( \frac{q-2}{q} \right), \end{equation*}
equivalently, 
\begin{equation}\label{nuestra}
R > \frac{r}{r+1}\left(1-\delta-\frac{2}{q+1}\right).
\end{equation}
\end{pro}

\begin{proof}

In the first examples we obtained $\delta = \dfrac{d}{n} = \dfrac{ \left(q^2/2 - 2q+3 \right) }{ q^2-q}$ and $R =\dfrac{k}{n} = \dfrac{ \left( q^2/2 -q +1 \right)}{q^2}.$

For the second, we obtained $\delta = \dfrac{d}{n} = \dfrac{\left(q^2 / 2  - 3q+3 \right)}{q^2-q}$ and $R = \dfrac{k}{n} = \dfrac{ \left( q^2/2 + 1 \right)}{q^2}$ .
Consequently, it verifies
\begin{align*} R + \frac{q-1}{q} \delta & =  \frac{q^2 - 3q + 4}{q^2}  >  \frac{q^2 - 3q + 2}{q^2} \\
& = \frac{ (q-1)(q-2)}{q^2} = \frac{q-1}{q} \left(  \frac{q-2}{q} \right) .\qedhere\end{align*}  \end{proof}

Finally, the relative parameters of the codes of Examples \ref{ej312} and \ref{ej313} also lie above the GV-bound.

\begin{pro}[Barg-Tamo] 
For sufficiently large $q$ and for certain $2 \leq i \leq q-1$ and $1 \leq l \leq (q-1)(q-i)$; we do have that the lower bound for the relative parameters $(\delta,R)$ of $C_i(S,D)$ with $D = l P_{\infty}$ do improve the bound described in \cite[Thm 6.1]{bargtamovladut}.
That is, for adequate $i,l$, we have that $(\delta, R)$ lies above the curve
	\begin{equation}\label{GVbound} R = \frac{r}{r+1} - \min_{0 < s \leq 1} \left\{	\frac{1}{r+1} \log_qb_2(s) -\delta \log_q(s)\right\}
\end{equation}
where 
$b_2(s) = \frac{1}{q} \left((1 + (q-1)s)^{r+1}+ (q-1)(1-s)^{r+1} \right).$
\end{pro}

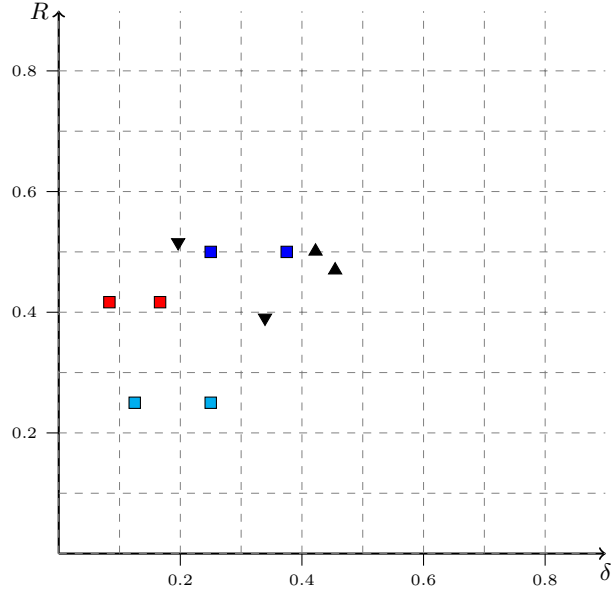
\begin{figure}[h]
\begin{tikzpicture}[scale = 8]
    \draw[thick,->] (0,0) -- (0.9, 0) node[anchor=north] {\small $\delta$};
    \draw[thick,->] (0,0) -- (0, 0.9) node[anchor=east] {\small $R$};
    \draw[step = 0.1, gray, dashed, very thin] (0,0) grid (0.899,0.899);
    \foreach \x in {0.2, 0.4, 0.6, 0.8} 
        \draw (\x,0) -- (\x,-0.02) node[below] {\tiny \x};
    \foreach \y in {0.2, 0.4, 0.6, 0.8} 
        \draw (0,\y) -- (-0.02,\y) node[left] {\tiny \y};
    
        \foreach \x/\y in {0.25/0.25,  0.125/0.25} {
        \node[draw=black, fill=cyan, shape=rectangle, minimum size=0.15cm, inner sep=0pt] at (\x,\y) {};}
       
        \foreach \x/\y in {0.25/0.5,  0.375/0.5} {
        \node[draw=black, fill=blue, shape=rectangle, minimum size=0.15cm, inner sep=0pt] at (\x,\y) {};}
        
        \foreach \x/\y in {0.16666666666666667/0.41666666666666667,  0.08333333333333333/0.41666666666666667} {
        \node[draw=black, fill=red, shape=rectangle, minimum size=0.15cm, inner sep=0pt] at (\x,\y) {};}
        
        \foreach \x/\y in {0.3392857142857143/0.390625, 0.19642857142857142/0.515625} {
        \node[draw=black, fill=black, regular polygon, regular polygon sides=3, minimum size=0.2cm, inner sep=0pt, rotate=180] at (\x, \y) {};}
        
        \foreach \x/\y in {0.45463709677419356 /  0.4697265625, 0.4223790322580645  / 0.5009765625} {
        \node[draw=black, fill=black, regular polygon, regular polygon sides=3, minimum size=0.2cm, inner sep=0pt, rotate=0] at (\x, \y) {};}
        
\end{tikzpicture}
 \caption{Comparison of some examples in the table \ref{tabla}} \label{ejemplotabla}
\end{figure}

\begin{figure}[ht]
\includegraphics[scale=0.4]{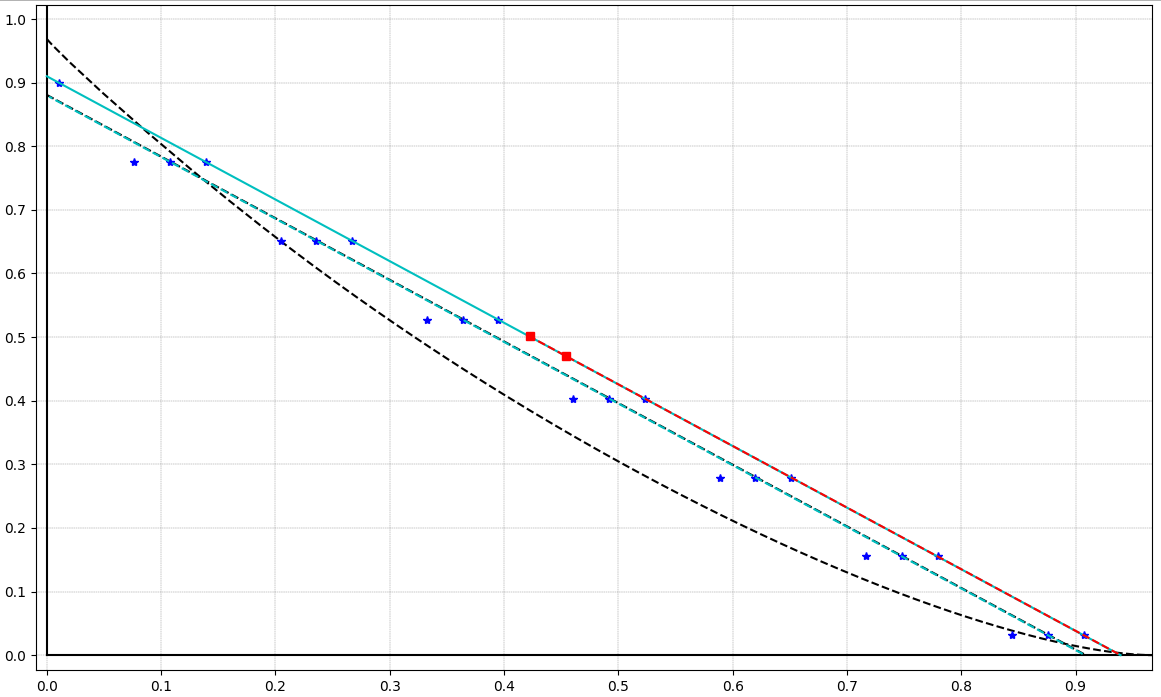}
\caption{Bounds \eqref{BTVbound} and \eqref{GVbound} for $r+1=q=32$ are shown as dashed curves in black.
For $2 \leq i \leq 4$, there are plotted lower bounds \textcolor{blue}{$C_2(\mathcal{B},V)$} as points in blue.
Ranges of constructible examples of  \textcolor{blue}{$C_3(\mathcal{B},V)$ obtained from Theorem \ref{teoremareversionado}} are shown as a dashed line in cyan, which lies strictly \textcolor{blue}{below} the $BTF$ bounds.
 The parameters of \textcolor{blue}{$C_2(\mathcal{B},V)$} from \ref{ej312} and \ref{ej313} are marked as red squares, lying on the line obtained from \eqref{nuestra}, \textcolor{blue}{and the parameters of  $C_2(\mathcal{B},V)$ derived from Corollary \ref{coro38} are shown as a dashed line in red}.
} \label{ejemploq32}
\end{figure}

\end{document}